\documentclass[aps,twocolumn,10pt,prl]{revtex4-1}
\usepackage{amsfonts}
\usepackage{amsmath}
\usepackage{amssymb}
\usepackage{amsthm}
\usepackage{array}
\usepackage{epsfig}
\usepackage{rotating,graphicx}
\usepackage{stix}
\usepackage{bm}%
\usepackage{verbatim}
\usepackage{pstricks,pst-node}
\usepackage{tensor}
\usepackage{tikz}
\usepackage{longtable}
\usepackage{textgreek}
\theoremstyle{plain}
\newtheorem{theorem}{Theorem}

\theoremstyle{definition}

\newtheorem{exatitle}{Example}
{\begin{exatitle} \label{#2} #1 \end{exatitle}}%
{\hfill $\Box$ \\}

\newlength{\figurewidth}
\newcommand{\ket}[1]{| #1 \rangle}
\newcommand{\bra}[1]{\langle #1 |}
\newcommand{\braket}[2]{\langle #1 | #2 \rangle}
\newcommand{\pket}[1]{[ #1 ]}

\newcommand{\td}{\text{d}}

\newcommand{\eg}{\hbox{\em e.g.{}}}
\newcommand{\etc}{\hbox{\em etc.{}}}
\newcommand{\ie}{\hbox{\em i.e.{}}}

\newcommand{\rhs}{\hbox{r.h.s.{}}}

\makeatletter
\g@addto@macro\bfseries{\boldmath}
\makeatother
\begin{document}

\title{Toponomic Quantum Computation}

\author{C.{} Chryssomalakos}
\email{chryss@nucleares.unam.mx}
\affiliation{Instituto de Ciencias Nucleares \\
	Universidad Nacional Aut\'onoma de M\'exico\\
	PO Box 70-543, 04510, CDMX, M\'exico.}

\author{E.{} Guzm\'an-Gonz\'alez}
\email{edgar.guzman@correo.nucleares.unam.mx}
\affiliation{Departamento de F\'{\i}sica \\
	Universidad Aut\'onoma Metropolitana-Iztapalapa\\
PO Box 55-534, 09340, CDMX, M\'exico.}

\author{L.{} Hanotel}
\email{hanotel@correo.nucleares.unam.mx}	
\affiliation{Instituto de Ciencias Nucleares \\
	Universidad Nacional Aut\'onoma de M\'exico\\
	PO Box 70-543, 04510, CDMX, M\'exico.}

\author{E.{} Serrano-Ens\'astiga}
\email{edensastiga@ens.cnyn.unam.mx}
\affiliation{Centro de Nanociencias y Nanotecnolog\'{\i}a,
Universidad Nacional Aut\'onoma de M\'exico\\
PO Box 14, 22800, Ensenada, Baja California, M\'exico
}

\begin{abstract}
	\noindent 
Holonomic quantum computation makes use of non-abelian geometric phases, associated to the evolution of  a subspace of quantum states, to encode logical gates. We identify a special class of subspaces, for which a sequence of rotations results in a non-abelian holonomy of a topological nature, so that it is invariant under any $SO(3)$-perturbation. Making use of a Majorana-like stellar representation for subspaces, we give explicit examples of  topological-holonomic (or \emph{toponomic})  NOT and CNOT gates.
\end{abstract}

\maketitle


It is well known that cyclic evolution of a quantum state $\pket{\psi}$ in the projective Hilbert space $\mathbb{P}(\mathcal{H}) \equiv \mathbb{P}$ of a physical system  gives rise to a geometric phase, which is invariant under time-reparametrizatons of the curve $C(t)$ traced in $\mathbb{P}$~\cite{Ber:84} (we denote by $\pket{\psi}$ points in $\mathbb{P}$, \ie, equivalence classes of normalized states $\ket{\psi}$ in the Hilbert space $\mathcal{H}$ that only differ  by an overall phase). The non-abelian generalization of this effect, due to Wilczek and Zee~\cite{Wil.Zee:84}, involves the cyclic evolution of a $k$-dimensional degenerate subspace $\Pi$ of the Hilbert space $\mathcal{H}$, in which case the phase factor mentioned above gets replaced by a unitary $k \times k$ matrix $U$ which acts in $\Pi$ --- as in the abelian case, $U$ is reparametrization invariant. This latter fact has prompted the suggestion of using such holonomies~\cite{Sim:83} in the implementation of logical gates in quantum computing, as it provides a certain level of immunity from noise effects~\cite{Zan.Ras:99,solinas2004robustness}.

Specifying the above to the case of a spin-$s$ system, so that $\mathcal{H}=\mathbb{C}^N$, $N=2s+1$, the simplest realization of a cyclic evolution of a state in $\mathbb{P}$ is through a sequence of rotations $R_t$ in $SO(3)$ (with $R_0$ being the identity operation), 
\begin{equation}
\label{psitR}
\pket{\psi_t}=\pket{D^{(s)}(\tilde{R}_t)\ket{\psi_0}}
\, ,
\end{equation}
as in the familiar example of a precessing spin, where $D^{(s)}(\tilde{R}_t)$ is the spin-$s$ irreducible representation of the lift $\tilde{R}_t$ of $R_t$ in $SU(2)$ --- said lift is unique, if one specifies that $\tilde{R}_0=I$ (rather than $-I$). For a general state $\pket{\psi_0}$, cyclicity implies that $R_t$ must start and end at the identity of $SO(3)$, $R_0=R_T=I$, but when $\pket{\psi_0}$ has a nontrivial discrete rotational symmetry group $\Gamma=\{R_{(0)}\equiv I, R_{(1)},\ldots,R_{(p)}\}$, $R_t$ might also  end at any of the symmetry rotations $R_{(i)} \in \Gamma$. In this latter case, it has been shown in~\cite{Agu.Chr.Guz.Han.Ser:19}, that when $\ket{\psi_0}$ is \emph{anticoherent}~\cite{Zim:06}, meaning its spin expectation value vanishes, the geometric phase acquired is of a topological character, and is thus invariant under arbitrary (not necessarily small) smooth perturbations. The principal aim of the present work is a generalization of these considerations to the case of the non-abelian holonomies mentioned above. To this end, we need to synthesize a number of results, some of them quite new.

The first ingredient we will need is a formulation of the Wilczek-Zee (WZ) effect, that parallels the treatment of the abelian case in~\cite{Muk.Sim:93}. In that reference, given a curve $\pket{\psi_t}$, $0 \leq t \leq T$, traced by a state  in $\mathbb{P}$, one chooses arbitrarily but smoothly a phase for each $t$, obtaining a lift  $\ket{\psi_t}$ in $\mathcal{H}$. In terms of the latter, the geometric phase assigned to $\pket{\psi_t}$ is given by
\begin{align}
\label{gpMS}
\varphi_\text{geo}
&=
\varphi_\text{tot}-\varphi_{\text{dyn}}
\nonumber
\\
&=
\arg \braket{\psi(0)}{\psi(T)}
+
i\int_0^T dt \, \bra{\psi(t)} \partial_t \ket{\psi(t)}
\, ,
\end{align}
where  the two terms  on the right hand side are known as the \emph{total} and \emph{dynamical} phase, respectively --- it is easily seen that $\varphi_\text{geo}$ does not depend on the particular lift $\ket{\psi_t}$ chosen, so it is a property of the curve $\pket{\psi_t}$ in $\mathbb{P}$.

The above physical quantum states $\pket{\psi}$ correspond to rays in $\mathcal{H}$, and the set of such rays is the projective space $\mathbb{P}$. Similarly, $k$-dimensional subspaces in $\mathcal{H}$, of the type that appear in the WZ effect, form the Grassmannian $\text{Gr}(k,N)$ (note that $\text{Gr}(1,N)=\mathbb{P}$) and the cyclic evolution of such a subspace corresponds to a closed curve  $\Pi(t)$, in $\text{Gr}(k,N)$, with $0 \leq t \leq T$ and $\Pi(0)=\Pi(T)$. Given an arbitrary orthonormal basis $\{\ket{\phi_i(t)}\}$, $i=1,\ldots,k$ in $\Pi(t)$, with $\ket{\phi_i(0)}$ not necessarily equal to $\ket{\phi_i(T)}$ (since a given plane may be spanned by many different bases), define the matrix $Q$ with entries $Q_{ij}=\braket{\phi_i(0)}{\phi_j(T)}$, and the Wilczek-Zee connection $\mathcal{A}(t)$, with entries $\mathcal{A}_{ij}(t)=\braket{\phi_i(t)}{\dot{\phi}_j(t)}$, the dot denoting time derivative. The polar part of $Q$ is given by $P=WV^\dagger$, where $Q=WDV^\dagger$ is the singular value decomposition of $Q$, with $W$, $V$ unitary, and $D$ diagonal. Then to the  curve $\Pi(t)$ we assign, following~\cite{Kul.Abe.Sjo:06} (see also~\cite{Sjo:15}), the unitary holonomy 
\begin{equation}
\label{WZholo}
U_\text{geo}=P F^{-1}
\, ,
\end{equation}
where 
\begin{equation}
\label{Fdef}
F=\text{Pexp}\left( \int_0^T \td t \mathcal{A}(t) \right)
\end{equation}
is the path-ordered exponential of the WZ connection --- this is clearly a nonabelian version of the exponentiation of~(\ref{gpMS}). It can be shown that $U_\text{geo}$ transforms covariantly under changes in the choice of the basis $\{\ket{\phi_i}\}$, and that it reproduces the standard WZ holonomy when applied to closed curves --- it is worth remarking  though that (\ref{WZholo}) may also be applied to open curves.

The second ingredient we will need is a generalization of the concept of anticoherent spin states. Consider a hermitean matrix $A$, acting on $\mathcal{H}$, and the action of an element $g$ of $SU(2)$ on it via conjugation,
\begin{equation}
\label{SU2act}
A \mapsto A'=D^{(s)}(g) A D^{(s)}(g)^{-1}
\, .
\end{equation}
A natural question to ask is whether this transformation mixes all hermitean $N \times N$ matrices, or whether there are subspaces of hermitean matrices that only transform among themselves. The answer  is that there is a 1-dimensional subspace (multiples of the identity matrix) that is invariant under this action, then a 3-dimensional subspace, spanned by matrices $T^{(s)}_{1,1}$, $T^{(s)}_{1,0}$, $T^{(s)}_{1,-1}$, that only transform among themselves, then a 5-dimensional subspace, \etc, up to a $(4s+1)$-dimensional invariant subspace --- the matrices $T^{(s)}_{\ell m}$, $\ell=0,1,\ldots,2s$, $-\ell\leq m \leq \ell$, are called \emph{spin-$s$ polarization tensors}~\cite{Var.Mos.Khe:88}. In particular, the triplet $T^{(s)}_{1m}$, $m=\pm 1,0$, are, up to normalization, the three $\mathfrak{su}(2)$ generators, $S_{\pm}$, $S_z$, in the spin-$s$ representation. 
We  define now a \emph{$t$-anticoherent spin-$s$ $k$-plane} $\Pi \in \text{Gr}(k,N)$ by the requirement that any basis $\{ \ket{\psi_i} \}$, $i=1,\ldots,k$ in $\Pi$ satisfies 
\begin{equation}
\label{anticohPi}
\bra{\psi_i}T^{(s)}_{\ell m}\ket{\psi_j}=0
\, ,
\end{equation}
for $1 \leq \ell \leq t$,
$
-\ell \leq m \leq \ell$,
$i,j=1,\ldots,k$.
Note that, just like for the $k=1$ definition, the above notion of $t$-anticoherence is rotationally invariant. We will actually only consider here 1-anticoherent planes, for which 
\begin{equation}
\label{oap}
\bra{\psi_i}\mathbf{S}^{(s)} \ket{\psi_j}=0
\, ,
\end{equation}
where $\mathbf{S}^{(s)}=(S^{(s)}_x,S^{(s)}_y,S^{(s)}_z)$ is the vector of the $\mathfrak{su}(2)$ generators, in the spin-$s$ representation --- this particular case can be related to the concept of anticoherent subspaces in~\cite{Per.Pau:17}.

Note that $SU(2)$ acts on a $k$-plane in $\mathcal{H}$ by transforming the elements of any basis that generates it, \ie, if $\Pi \in \text{Gr}(k,N)$ is generated by $\{\ket{\phi_i} \}$, $i=1,\ldots,k$, then the transformed plane $\Pi_g$, $g \in SU(2)$, is generated by $\{D^{(s)}(g)\ket{\phi_i} \}$, $i=1,\ldots,k$. Of crucial importance to our subsequent discussion is the orbit $\mathcal{O}_\Pi$  of $\Pi$ under  this action, $\ie$, the set of all $k$-planes that can be obtained by transforming $\Pi$. For a generic plane $\Pi$, with $k \geq 2$, \ie, one that has no rotational symmetries, the orbit $\mathcal{O}_\Pi$ is essentially a copy of $SO(3)$ since distinct $SU(2)$ elements give rise to distinct planes, except when they differ only by a sign. In the presence of a discrete rotational symmetry group $\Gamma$, however, two rotations $R$ and $R'$ that differ by a symmetry, $R'=R R_{(i)}$ with $R_{(i)} \in \Gamma$, give identical results when applied to $\Pi$ (via their lift to $SU(2)$), and $\mathcal{O}_\Pi$ is then identified with the quotient space $SO(3)/\Gamma$. 

The above two ingredients may be brought together in the following
\begin{theorem}
\label{thWZhol}
Let $\Pi$ be a 1-anticoherent spin-$s$ $k$-plane, with non-trivial rotation symmetry group $\Gamma=\{R_0=I,R_1,\ldots,R_p\}$ and denote by $\Pi(t)=R(t)(\Pi)$ the curve in $\text{Gr}(k,N)$  obtained by a sequence of rotations $R(t) \in SO(3)$ of $\Pi$, that starts, at $t=0$, at the identity, and ends, at $t=T$, on some symmetry rotation $R_m \in \Gamma$, so that $\Pi(T)=\Pi(0)$.
The WZ holonomy associated, via~(\ref{WZholo}), to the above closed curve  only depends on its homotopy class in $\mathcal{O}_\Pi$, and is therefore invariant under continuous deformations that fix its endpoints.
\end{theorem}
\begin{proof}
Let $\{\ket{\phi_i}\}$, $i=1,\ldots,k$, be a basis in $\Pi=\Pi(0)$, then $\Pi(t)$ is generated, by definition, by the kets $\{\ket{\phi_i(t)}=D^{(s)}(R(t))\ket{\phi_i} \}$, where, for $R(t)=R_{\mathbf{m}(t)}$,
\begin{equation}
\label{Dsdef}
D^{(s)}(R(t))=e^{-i \mathbf{m}(t) \cdot \mathbf{S}^{(s)}}
\, ,
\end{equation}
the direction of $\mathbf{m}(t)$ defining the rotation axis, while its modulus, in the interval $[0,\pi]$, the corresponding rotation angle. Then we get for the time derivative $\ket{\dot{\phi}_j(t)}=-i \mathbf{m}'(t) \cdot \mathbf{S}^{(s)}\ket{\phi_j(t)}$, where $\mathbf{m}'(t)$ is given by a complicated expression, the details of which are immaterial to the present argument, since
we may already conclude that the WZ connection vanishes for all times,
\begin{equation}
\label{Aijnull}
\mathcal{A}_{ij}(t)=-i \bra{\phi_i(t)}  \mathbf{m}'(t) \cdot \mathbf{S}^{(s)} \ket{\phi_j(t)}=0
\, ,
\end{equation}
by virtue of the 1-anticoherence assumption about $\Pi$, and the rotational invariance of the concept. Then the second factor in the \rhs{} of~(\ref{WZholo}) is equal to the unit matrix, and $U_\text{g}=Q_P$, which, evidently, is invariant under continuous deformations of the curve $\Pi(t)$ that fix its endpoints.
\end{proof}

The above theorem  forms the basis for the applications to quantum computing we discuss later on. Before getting to that part though, there is still one final ingredient missing that needs to be incorporated in our discussion, and which is captured in the following two related questions:
\begin{enumerate}
\item
How do we identify closed curves in the Grassmannian?
\item
Given a spin-$s$, $k$-plane in $\mathcal{H}$, generated by the kets $\{ \ket{\phi_i}\}$, $i=1, \ldots, k$, how do we identify its possible rotational symmetries?
\end{enumerate}
Both questions are nontrivial, because we have been specifying a plane giving a basis in it, and there are many bases generating the same plane. Thus, for a curve $\Pi(t)$ in the Grassmannian, with the plane $\Pi(t)$ being generated by the kets $\{ \ket{\phi_i(t)}\}$, $i=1, \ldots, k$, one may have $\ket{\phi_i(T)} \neq \ket{\phi_i(0)}$, and, yet, $\Pi(T)=\Pi(0)$. Similarly, the  basis kets rotated by a particular rotation $R_m$ might not coincide with the original ones, $D^{(s)}(R_m) \ket{\phi_i} \neq \ket{\phi_i}$, but the plane they generate after the rotation might be identical with the original one, $R_m(\Pi)=\Pi$. Clearly, to be able to describe efficiently loops in $\text{Gr}(k,N)$, we need an intrinsic characterization of a $k$-plane, that brings to the forefront its transformation properties under rotations. 
For the case $k=1$, \ie, when dealing with physical states in $\mathbb{P}$, there is such a characterization, due to Majorana~\cite{Maj:32} (see also~\cite{Ben.Zyc:17,Chr.Guz.Ser:18}). It is well known that spin-1/2 states can be represented as points on the Bloch sphere, essentially via their spin expectation value, 
\begin{equation}
\label{ponBs}
\pket{\psi} \mapsto \hat{n}=2\bra{\psi}\mathbf{S}\ket{\psi}
\, .
\end{equation}
It is also a fact that the quantum states of a system of $2s$ qubits, that are symmetric under permutations of the qubits, are in 1-to-1 correspondence with spin-$s$ states, \eg, the 2-qubit symmetric state $(\ket{+-}+\ket{-+})/\sqrt{2}$ corresponds to the spin-1 state $\ket{s=1,m=0}$ --- a similar correspondence holds for \emph{any} spin-$s$ state $\ket{\psi}$,
\begin{equation}
\label{psisqubits}
\ket{\psi} \mapsto \ket{\psi_1} \vee \ket{\psi_2} \vee \ldots \vee \ket{\psi_{2s}}
\, ,
\end{equation}
where the $\ket{\psi_i}$ are qubit states and the $\vee$-product denotes a symmetrized tensor product, \eg, $\ket{\psi_1}\vee \ket{\psi_2} \equiv \ket{\psi_1} \otimes \ket{\psi_2} +\ket{\psi_2} \otimes \ket{\psi_1})/\sqrt{2}$. Of course, a 1-to-1 correspondence can always be established, in an infinity of ways,  between vector spaces of equal dimension --- the special property of the particular correspondence mentioned above is that it commutes with the action of rotations, namely, given a spin-$s$ state $\ket{\psi}=(\psi_1,\ldots,\psi_N)$, one may rotate it by $R \in SO(3)$ by left-multiplying it with  $D^{(s)}(R)$, or, one may represent $\ket{\psi}$ as in~(\ref{psisqubits}), transform each of the qubits with the spin-1/2 matrix $D^{(1/2)}(R)$, and map the result back into a spin-$s$ state with the inverse mapping --- the results of the two sequences of operations are the same. Since each state in the \rhs{} of~(\ref{psisqubits}) can be represented by a point on the unit sphere, $\ket{\psi}$ itself can be represented by an unordered (because of the symmetrization) set of $2s$ points (\emph{stars}) on the sphere, its \emph{Majorana constellation}. When $\ket{\psi}$ is transformed by $D^{(s)}(R)$ in $\mathcal{H}$, its constellation rigidly rotates by $R$ in physical space. To appreciate the enormous advantage furnished by this representation of rays in $\mathcal{H}$, consider the spin-2 state 
\begin{equation}
\label{psimystery}
\ket{\chi}
=
\left(
\frac{\sqrt{3}-2 \sqrt{6}}{12} ,
\frac{\sqrt{3}+\sqrt{6}}{6},
\frac{1}{2 \sqrt{2}},
\frac{\sqrt{3}-\sqrt{6}}{6},
\frac{4+\sqrt{2}}{4\sqrt{6}}
\right)
\, ,
\end{equation}
and ponder whether it has any rotational symmetries, \ie, whether there exists any rotation $R$ such that $D^{(2)}(R) \ket{\chi}=e^{i\alpha} \ket{\chi}$, $\alpha \in \mathbb{R}$, so that $\pket{\chi} \in \mathbb{P}$ is invariant under $R$ --- admittedly not an easy question. But then, taking a look at the corresponding Majorana constellation, shown in Fig.~\ref{tetrarot:Fig}, one realizes (\eg, by verifying that all stars are equidistant) that it forms a regular tetrahedron, and one immediately concludes that 
$\pket{\chi}$ is invariant, \eg, under a rotation about the $x$-axis by $2\pi/3$. 
\setlength{\figurewidth}{\textwidth}
\begin{figure}
\raisebox{0\totalheight}{%
\includegraphics[angle=0,width=.35\figurewidth]%
{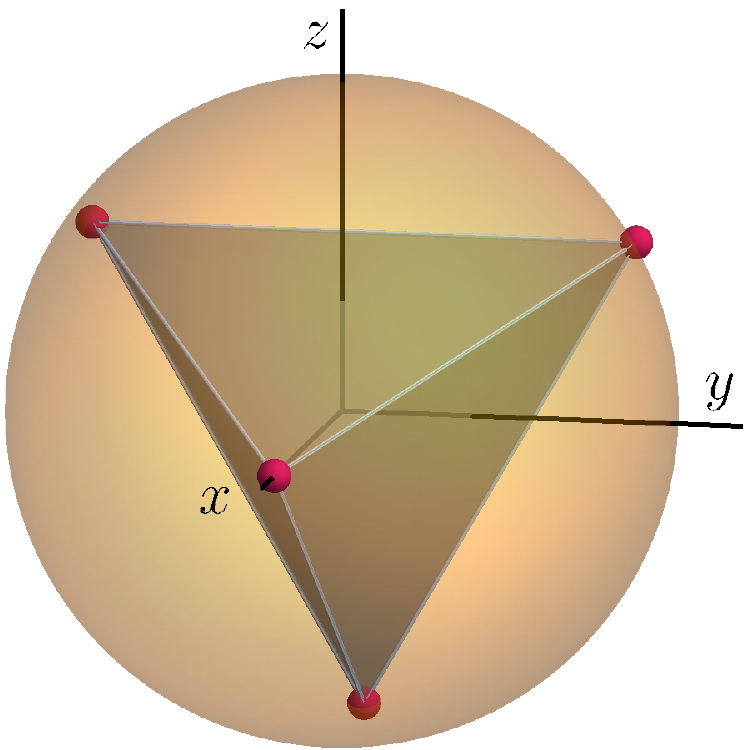}%
}%
\caption{Majorana constellation of the spin-2 state $\ket{\chi}$ of Eq.~(\ref{psimystery}), forming a regular tetrahedron with one vertex on the $x$-axis.}
\label{tetrarot:Fig}
\end{figure}

Could it be that a similar stellar representation exists for spin-$s$ $k$-planes? Indeed it does, as it has been recently shown in~\cite{Chr.Guz.Han.Ser:21}, albeit in a more complicated form. A general spin-$s$ $k$-plane can be uniquely represented by a family of constellations, of various spins, the values of which depend on both $s$ and $k$, each of which is weighted by a complex number. If one arbitrarily orders the constellations, one may treat these complex weights as a pseudo-spinor, and assign to it a ``spectator'' constellation, so that the above $k$-plane representation is purely visual. There is a way of doing this such that under rotations, the Majorana constellations rotate rigidly, while the spectator constellation remains invariant. Thus, any rotational symmetries of the $k$-plane are conveniently encoded in its multiconstellation.

As a concrete example of the above general procedure, consider the spin-2 2-plane $\Pi_{\text{NOT}} \in \text{Gr}(2,5)$, generated by 
\begin{align}
\label{PiNOT1}
\ket{\psi_1}
&=
\frac{1}{\sqrt{3}} (1,0,0,\sqrt{2},0)
\\
\label{PiNOT2}
\ket{\psi_2}
&=
\frac{1}{\sqrt{3}} (0,-\sqrt{2},0,0,1)
\, .
\end{align}
The Majorana constellations of these states are both regular tetrahedra, antipodal to each other, with one vertex on the $z$-axis (see image on the left in Fig.~\ref{PiNOT:Fig}). Following the procedure detailed in~\cite{Chr.Guz.Han.Ser:21}, we find that the vector space inside which the spin-2 2-planes are naturally embedded (this is the \emph{Pl\"ucker embedding}, see, \eg, p.{} 43 of~\cite{Sha:13}, or p.{} 110 of~\cite{Jac:10}) is 10 dimensional, and splits into  spin-3 and spin-1 subspaces. For   $\Pi_{\text{NOT}}$ itself we find the components
\begin{equation}
\label{PiNOTcomp}
\Pi_{\text{NOT}}=\frac{1}{3}\left( \big(-\sqrt{2},0,0,\sqrt{5},0,0,\sqrt{2} \big) \big(0,0,0 \big) \right)
\, ,
\end{equation}
where we have grouped together the components of each spin multiplet. Since the spin-1 component vanishes in this case, the multiconstellation of $\Pi_{\text{NOT}}$ only has a spin-3 constellation, shown on the right in Fig.~\ref{PiNOT:Fig} --- this is a regular octahedron, which fully represents the rotational symmetries of the 2-plane.  Accordingly, the pseudo-spinor of relative weights is $(1,0)$, and the spectator constellation, interpreting the pseudo-spinor as a spin-1/2 state, consists of a single star at the north pole. 

It is easily checked that $\Pi_\text{NOT}$ is 1-anticoherent (as defined in~(\ref{anticohPi})), so our theorem~\ref{thWZhol} may be used. 
We give an example of a cyclic evolution, involving a symmetry rotation, and producing (as WZ holonomy) a toponomic logical NOT gate.
Consider, to that effect, a sequence of rotations $R(t)=R_{\pi t \hat{\mathbf{y}}}$ applied to $\Pi_{\text{NOT}}$, that starts, at $t=0$, at the identity, and ends, at $t=1$, on the symmetry rotation $R(1)=R_{(0,\pi,0)}$, \ie, around the $y$-axis by $\pi$. The WZ holonomy for this curve is computed, from~(\ref{WZholo}), to be 
$\sigma_x$, \ie, the logical gate NOT, and this result is invariant under continuous, however large, perturbations of $R(t)$ --- accordingly, this realization of the NOT gate is totally immune to noise.
\setlength{\figurewidth}{\textwidth}
\begin{figure}
\raisebox{0\totalheight}{%
\includegraphics[angle=0,width=.23\figurewidth]%
{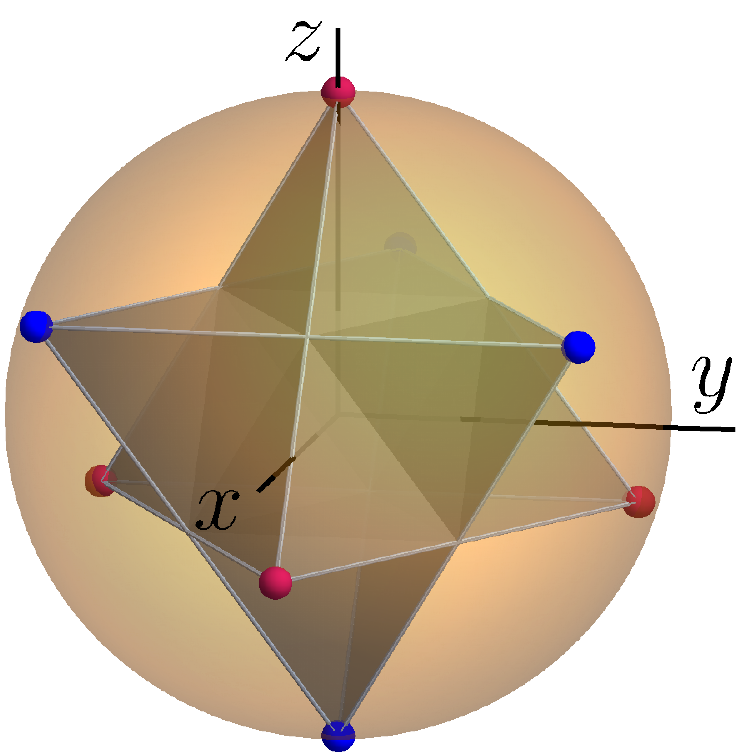}%
}
$\phantom{\,}$
\raisebox{0\totalheight}{%
\includegraphics[angle=0,width=.23\figurewidth]%
{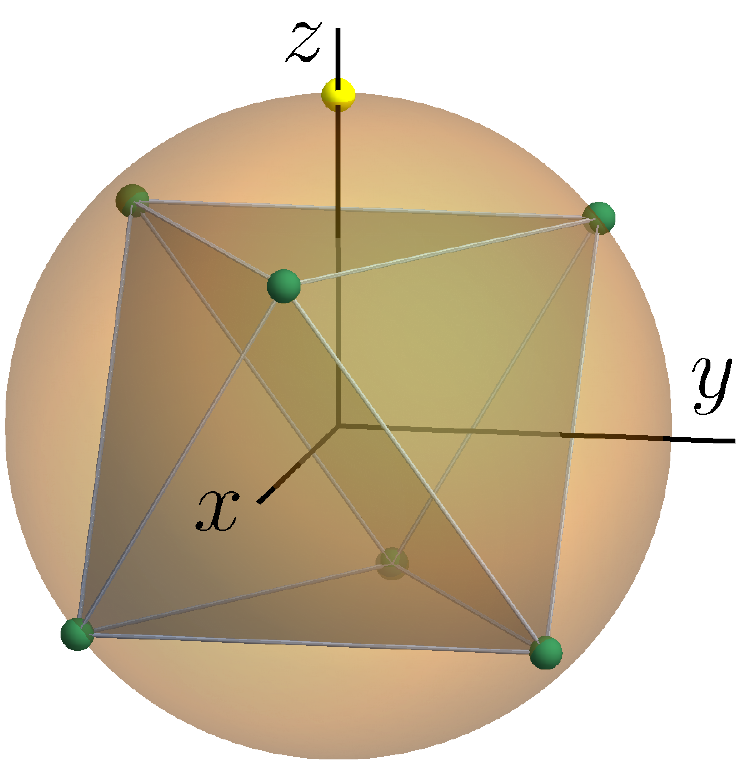}%
}%
\caption{Left: Majorana constellations of the states $\ket{\psi_1}$ (red), $\ket{\psi_2}$ (blue) --- see Eqs.~(\ref{PiNOT1}), (\ref{PiNOT2}). Right: Majorana-like spin-3 constellation of the plane $\Pi_{\text{NOT}}$, generated by $\ket{\psi_1}$, $\ket{\psi_2}$ --- it forms a regular octahedron. The spin-1 component of $\Pi_{\text{NOT}}$ is zero, hence, its spectator constellation consists of a single star at the north pole (yellow).}
\label{PiNOT:Fig}
\end{figure}

It is worth emphasizing that we were able to identify the above  symmetry rotation of $\Pi_{\text{NOT}}$ by inspection of its multiconstellation. It is true that the particular basis $\ket{\psi_{1,2}}$ used already makes this symmetry obvious, but it should be kept in mind that there is no guarantee that each of the elements of an arbitrary basis of a plane $\Pi$ possesses the rotational symmetries of $\Pi$. For example, consider the alternative basis of $\Pi_{\text{NOT}}$
\begin{align*}
\ket{\psi_a}
&=
\cos\frac{\pi}{3} \ket{\psi_1}+\sin\frac{\pi}{3} \ket{\psi_2}
=
\frac{1}{\sqrt{6}} (\frac{1}{\sqrt{2}},-\sqrt{3},0,1,\frac{\sqrt{6}}{2})
\\
\ket{\psi_b}
&=
\cos\frac{\pi}{6} \ket{\psi_1}-\sin\frac{\pi}{6} \ket{\psi_2}
=
\frac{1}{\sqrt{6}} (\frac{\sqrt{6}}{2},1,0,\sqrt{3},-\frac{1}{\sqrt{2}})
\, ,
\end{align*}
with the corresponding constellations shown in Fig.~\ref{PiNOT2:Fig} --- neither of the basis states has the symmetry $R_{(0,\pi,0)}$ of $\Pi_{\text{NOT}}$.
\setlength{\figurewidth}{\textwidth}
\begin{figure}
\raisebox{0\totalheight}{%
\includegraphics[angle=0,width=.23\figurewidth]%
{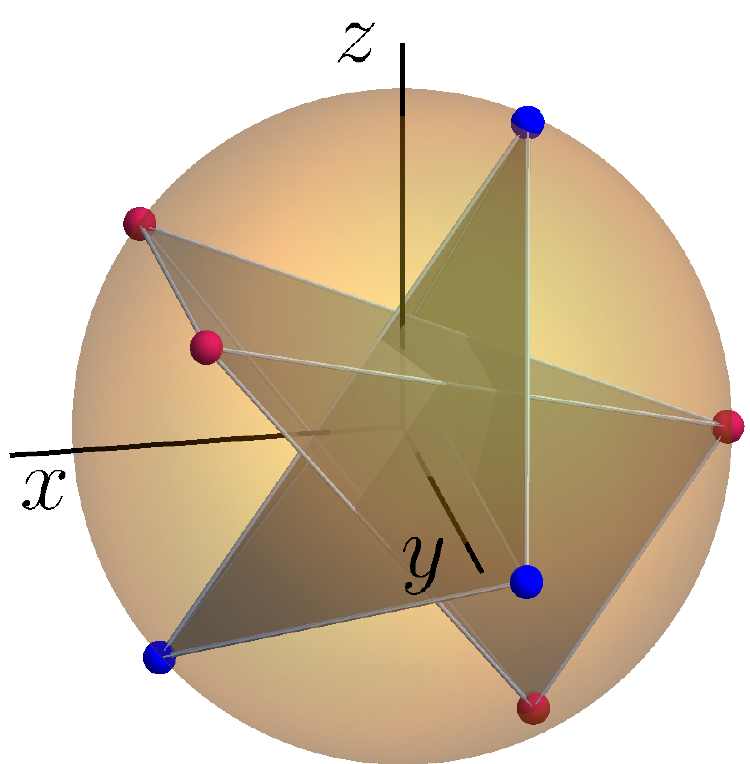}%
}
\caption{Majorana constellations of the states $\ket{\psi_{a}}$ (red), $\ket{\psi_b}$ (blue), which form an alternative basis of $\Pi_{\text{NOT}}$.} 
\label{PiNOT2:Fig}
\end{figure}

As a second example, consider the spin-5 4-plane $\Pi_{\text{CNOT}}$, generated by the states
\begin{align}
\ket{\psi_1} &= \frac{1}{2}( \ket{5,5} + \sqrt{2} i \ket{5,0} +  \ket{5,-5} ) 
\, ,
\label{s5k4a}
\\
\ket{\psi_2} &=  \frac{1}{2}( \ket{5,5} - \sqrt{2} i \ket{5,0} +  \ket{5,-5} ) 
\, ,
\label{s5k4b}
\\
\ket{\psi_3} &= \frac{1}{\sqrt{5}}( \sqrt{2}\ket{5,3} + \sqrt{3} i \ket{5,-2} ) 
\, ,
\label{s5k4c}
\\
\ket{\psi_4} &= \frac{1}{\sqrt{5}}( \sqrt{3} i\ket{5,2} + \sqrt{2} \ket{5,-3} ) 
\, ,
\label{s5k4d}
\end{align}
which can be shown to be 1-anticoherent. An analysis similar to that presented above allows the identification of its rotational symmetries --- space limitations do not allow us to give the full multiconstellation, we only show the principal (\ie, highest-spin) constellation in Fig.~\ref{PrCo2:Fig}.
\setlength{\figurewidth}{\textwidth}
\begin{figure}
\raisebox{0\totalheight}{%
\includegraphics[angle=0,width=.23\figurewidth]%
{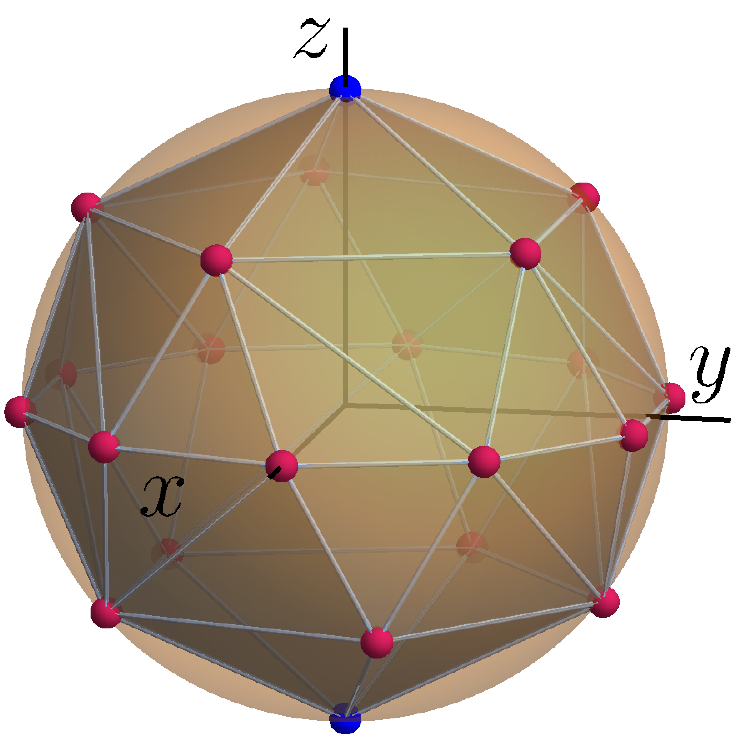}%
}
\caption{Principal constellation of the spin-5 4-plane in~(\ref{s5k4a})-(\ref{s5k4d}). The north and south pole, shown in blue (darker gray in print), are quadruply occupied.} 
\label{PrCo2:Fig}
\end{figure}
We find that the sequence of rotations $R_1(t)=R_{\pi t \hat{\mathbf{x}}}$ ends, at $t=1$, on the symmetry rotation $R_{(\pi,0,0)}$ of $\Pi_{\text{CNOT}}$, with the corresponding holonomy $U_1$ being given on the left in the equation below,
\begin{equation}
U_1=-\left(
\begin{array}{cccc}
1 & 0  & 0 & 0
\\
0 & 1 & 0 & 0
\\
0 & 0 & 0 & 1
\\
0 & 0 & 1 & 0
\end{array}
\right)
\, ,
\quad
U_2=
\left(
\begin{array}{cccc}
 1 & 0 & 0 & 0 \\
 0 & 1 & 0 & 0 \\
 0 & 0 & e^{i 4 \pi/5} & 0 \\
 0 & 0 & 0 & e^{-i 4 \pi /5} \\
\end{array}
\right)
\end{equation}
\ie, it is a CNOT gate, with an extra overall sign, which is, again, totally immune to noise, while $U_2$ on the right above is the holonomy associated to the sequence of rotations  $R_2(t)=R_{2\pi t \hat{\mathbf{z}}/5}$, $0 \leq t \leq 1$, with $R_2(1)$ another symmetry rotation of $\Pi_{\text{CNOT}}$. 

Concluding this short letter we point out the necessity to refine our search for anticoherent planes in the spin-$s$ Hilbert space. The ones presented above, as well as several others that we are aware of,  were found with a variety of \emph{ad hoc} methods that fail to convey a satisfactory geometrical picture of their locus inside the corresponding Grassmannian --- we are currently pursuing such an understanding in the hope of realizing additional noise-tolerant logical gates.

\section*{Acknowledgements}
\label{Ack}
CC, LH, and ESE would like to acknowledge partial financial support from the DGAPA PAPIIT project IN111920 of UNAM. ESE would also like to thank DGAPA UNAM for a postdoctoral fellowship.

\end{document}